\documentclass{article}
\usepackage[utf8]{inputenc}

\title{Interacting SPT phases are not Morita invariant}
\author{Luuk Stehouwer}
\date{\today}

\usepackage{amsmath}
\usepackage{amsfonts}
\usepackage{amssymb}
\usepackage{amsthm}
\usepackage{enumerate}
\usepackage{color}
\usepackage{hyperref}
\usepackage{graphicx}
\usepackage{rotating}
\usepackage{stmaryrd}

\usepackage{tikz-cd}

\newcommand{\Z}{\mathbb{Z}}

\newcommand{\C}{\mathbb{C}}
\newcommand{\R}{\mathbb{R}}

\newcommand{\Hom}{\operatorname{Hom}}

\newcommand{\pt}{\operatorname{pt}}

\newcommand{\Spin}{\operatorname{Spin}}
\newcommand{\Pin}{\operatorname{Pin}}

\newcommand{\nw}{\rotatebox[origin=c]{180}{$w$}}
\newcommand{\np}{\rotatebox[origin=c]{180}{$q$}}

\newtheorem{proposition}{Proposition}
\newtheorem{theorem}[proposition]{Theorem}

\newtheorem{corollary}[proposition]{Corollary}
\newtheorem*{result}{Main Result}

\theoremstyle{definition}

\newtheorem{definition}[proposition]{Definition}
\newtheorem*{ansatz}{Ansatz}

\theoremstyle{remark}

\newtheorem{example}[proposition]{Example}

\begin{document}

\maketitle

    \begin{abstract}
    The tenfold way provides a strong organizing principle for invertible topological phases of matter.
Mathematically, it is intimately connected with $K$-theory via the fact that there exist exactly ten Morita classes of simple real superalgebras.
This connection is physically unsurprising, since weakly interacting topological phases are classified by $K$-theory.
We argue that when strong interactions are present, care has to be taken when formulating the exact ten symmetry groups present in the tenfold way table.
We study this phenomenon in the example of class D by providing two possible mathematical interpretations of a class D symmetry.
These two interpretations of class D result in Morita-equivalent but different symmetry groups.
As $K$-theory cannot distinguish Morita-equivalent protecting symmetry groups, the two approaches lead to the same classification of topological phases on the weakly interacting side.
However, we show that these two different symmetry groups yield different interacting classifications in spacetime dimension 2+1.
We use the approach to interacting topological phases using bordism groups, reducing the relevant classification problem to a spectral sequence computation.

\end{abstract}

\section{Introduction}

A symmetry-protected topological (SPT) phase is a material which stays in a robust nontrivial state at zero temperature because it is protected by a symmetry \cite{senthil2015symmetry}.
In recent years several mathematical approaches to SPT-phases of matter have been developed.
With the simplifying assumption of free particles, the classification of fermionic SPTs by $K$-theoretic methods has been well-established \cite{kitaev2009periodic} \cite{freedmoore}.
Given the prominent role $8$- and $2$-Bott periodicity play in real respectively complex $K$-theory, it is not surprising that there is a corresponding important phenomenon in free fermion SPTs.
This is called the tenfold way, which is a collection of ten classes in which physical systems with topological properties can fall.
However, there are several competing mathematical formalisms to describe the tenfold way.
For example, some formulations of the tenfold way require Hamiltonians to always be charge-conserving \cite{thiangk-theoretic}, while others allow Bogliubov-de-Gennes Hamiltonians \cite{altlandzirnbauer} \cite{kitaev2009periodic}.

In this paper, a general mathematical framework of fermionic symmetry groups $G$ is developed with the goal to describe what it means for a topological phase to be protected by an internal symmetry group $G$.
One version of the ten-fold way is then outlined as a collection of ten special fermionic groups that correspond to the ten symmetry classes in agreement with \cite{freedhopkins}.
Using the language of fermionic groups, it can be argued that there are other choices of internal symmetry groups that also realize the tenfold way in some sense.
Namely, it turns out that such different approaches to the ten-fold way give equivalent classifications of topological phases in the weakly interacting approximation.
The reason is that the $K$-theory of a symmetry group only depends on the representation theory and therefore different symmetry groups with the same representation theory yield isomorphic groups of free fermion SPT-phases.
This is called Morita invariance and justifies that we can talk about \emph{the} ten-fold way in the free setting.
Morita invariance is briefly treated in section \ref{subsec:Morita equiv}, while further mathematical details will be justified in future work \cite{stehouwertenfold}. 
However, Bott periodicity is a $K$-theoretic phenomenon and SPT-phases involving strongly interacting particles are no longer expected to be classified by $K$-theory.
Therefore, it is not too surprising that one has to be more careful implementing the tenfold way in the strongly interacting setting.
In particular, it is essential for computations with strong interactions to know the exact symmetry groups of the system in the ten cases.

However, for non-experts it can be challenging to find consensus in the literature about what the symmetry groups are of the ten cases and therefore some caution is needed when talking about `the symmetry group of class X' for interacting topological phases.
We illustrate this issue by focusing on class D topological superconductors.
The standard way to realize class D is by requiring no symmetries at all and so in particular no charge conservation, see for example the original work of Altland-Zirnbauer \cite{altlandzirnbauer} and Kitaev \cite{kitaev2009periodic}. 
This corresponds to the choice of fermionic internal symmetry group $\Z_2^F = \{1, (-1)^F\}$.
Another description is that class D SPTs have a particle-hole symmetry $C$ which is complex antilinear on one particle Hilbert space and satisfies $C^2 = 1$.
As emphasized in \cite{zirnbauerparticlehole}, there seems to be no agreement in the literature on the definition of a particle-hole symmetry.
Therefore to interpret this other description, it is essential to clarify its possible definitions.
Some authors work on the space of Nambu spinors and use the term `particle-hole symmetry' to refer to the canonical correspondence between particles and holes on Nambu space referred to as `particle-hole conjugation' in loc. cit. 
In that case, a condensed matter theorist saying that there is a particle-hole symmetry is equivalent to a high-energy physicist saying that there is no symmetry.
Indeed, this particle-hole conjugation is always present (also for class A for example) and is not a symmetry in the sense that there is no canonical way to second quantize it to non-trivial symmetry on Fock space.
Another interpretation of the statement that class D has a particle-hole symmetry $C$, is to use the definition of particle-hole symmetry of loc. cit. where it is defined as a symmetry that anti-commutes with a charge symmetry $Q$.
This results in a different internal symmetry group which we write $U(1)_Q \rtimes \Z_2^C$.
In the weakly interacting setting, it can be argued that `particle-hole symmetry cancels charge' in a suitable sense.
Namely, the classification results are independent of which of the two internal symmetry groups are taken because of a Morita equivalence.
We will elaborate on these statements and provide clear definitions in section \ref{subsec:same class}.

The main goal in the second part of this paper is to show that in the approach to interacting SPT-phases using bordism and invertible field theories of \cite{freedhopkins} that we recall in section \ref{sec: interacting}, the two symmetry groups give different classifications.
This paper therefore sheds light on the mechanism of `breaking the electromagnetic group $U(1)$ to $\Z_2 = \{\pm 1\}$ through a particle-hole symmetry', see the discussion in the first paragraph of \cite[section 9]{freedhopkins}.
We do this by computing the relevant bordism groups using techniques from algebraic topology.
The group of invertible field theories with internal symmetry $\Z_2^F$ is well-known and can be easily derived from the spin bordism groups \cite[Corollary 9.81]{freedhopkins}.
We deploy a spectral sequence to compute SPT-phases for the other symmetry group $U(1)_Q \rtimes \Z_2^C$ in low dimensions. 
We show in section \ref{sec: computing the two} that in spacetime dimension $2+1$ there are both a $\Z_2$-invariant and a $\Z$-invariant for this internal symmetry group which does not appear in the other classification.
With knowledge of this result, an explicit representative invertible topological field theory with nontrivial $\Z_2$-invariant is constructed, resulting in the partition function given in equation \ref{eq:partition function}.
Finally in section \ref{sec:explicit}, we give a spacetime on which this partition function is nontrivial.

\begin{result}
The classifications of interacting SPT-phases using bordism groups for the internal symmetry groups $\Z_2^F$ containing only fermion parity and $U(1)_Q \rtimes \Z_2^C$ containing both charge and a complex-linear particle hole symmetry with $C^2 =1$ are not equal in dimension $2+1$:
\[
SPT^3_{\Z_2^F}  = \Z \quad SPT^3_{U(1)_Q \rtimes \Z_2^C} = \Z_2 \oplus \Z^2.
\]
\end{result}

Possible future research could be extending this result to other symmetry classes which have been described by different internal symmetry groups such as class AI (see section \ref{subsec:same class}).
Another topic is the map from free to interacting SPT phases, which conjecturally is given by a certain version of the Atiyah-Patodi-Singer $\eta$-invariant \cite[conjecture 10.25]{freedlectures} \cite[remark 9.72]{freedhopkins} \cite[section 9.2.6]{freedhopkins}.
With the choice of symmetry group $\Z_2^F$, it is known to be an isomorphism \cite[Corollary 9.81]{freedhopkins}.
But other examples show different behavior, such as the time-reversal invariant Kitaev chain in dimension $1+1$ for which it is a surjection $\Z \to \Z_8$. 
In particular, for the symmetry group $U(1)_Q \rtimes \Z_2^C$, it could be some interesting homomorphism $\Z \to \Z_2 \oplus \Z^2$.

We will make one remark on terminology. 
In this paper the phrase `SPT-phase' should be understood to mean `invertible phase'.
In other words, we will not make the common assumption that the phase becomes trivial after the symmetry is broken.
A translation to the other convention is easily made by taking the kernel of the map $SPT_G \to SPT_{\Z_2^F}$ given by forgetting the internal symmetry $G$.

\subsection*{Acknowledgements}

I offer many thanks to my advisor Peter Teichner. 
Not only did he teach me about the James spectral sequence and checked my computations, many of the ideas on fermionic symmetry groups and the tenfold way of Freed-Hopkins presented here are originally due to him.
I also thank the anonymous reviewer for many useful comments and helpful discussions.
Additionally, this paper benefited from discussions with Lukas M\"uller, Bernardo Uribe and Martin Zirnbauer.
I am grateful to the Max Planck Institute for Mathematics in Bonn where this research was carried out.

\section{Fermionic Symmetries and Symmetry Classes}

\subsection{Fermionic symmetry groups}
\label{sec: fermionic symmetry}

In this section we will abstract some properties of symmetries of physical systems involving fermions to arrive at the notion of a fermionic symmetry group.
Since the work of Wigner \cite{wigner1931gruppentheorie} it has become clear that in quantum-mechanics time-reversing symmetries $g$ have to be treated differently from time-preserving symmetries; they are anti-unitary instead of unitary operators.
Also in quantum field theories in spacetime dimension $d>1$, time-reversing symmetries act anti-unitarily on the many-body Hilbert space and so their representation theory behaves differently.
Additionally, symmetry groups of fermionic systems have an important datum that one has to keep track of: the fermion number operator $(-1)^F$.
We will assume all symmetries of our fermionic system are bosonic, so they commute with $(-1)^F$.
This discussion results in the following formal mathematical definition.

\begin{definition}
\label{def: fermionic group}
A \emph{fermionic (symmetry) group} is a topological group $G$ together with a central element $(-1)^F \in G$ of square one called \emph{fermion parity} and a continuous homomorphism $\theta: G \to \Z_2 = \{0,1\}$ such that $\theta((-1)^F) = 0$.
We write $\Z_2^F \subseteq G$ for the $\Z_2$-subgroup\footnote{We adopt the physicist's convention in this paper and write sub- and superscript letters on groups to remind ourselves of the physical meaning the groups play.} generated by $(-1)^F$.
\end{definition}

The bosonic group $G_b$ of a fermionic group $G$ is defined to be $G/\Z_2^F$.
We write $G_{ev} := \ker \theta$ for the symmetries that preserve the direction of time.

\begin{example}[class DIII]
\label{example: class DIII}
The fermionic group denoted $\Z_4^{T,F}$ is the group $\Z_4$ generated by an element denoted $T$ (time-reversal).
We take both $\theta$ and $(-1)^F$ to be nontrivial, so that $\theta(T) = 1$ is time-reversing and $T^2 = (-1)^F$.
\end{example}

\begin{example}
\label{example: Pin}
Let $G = \Pin^-(d)$ denote the $\Pin$-group; the nontrivial double cover of $O(d)$ in which lifts of reflections in $O(d)$ have negative squares.
It has a canonical square one central element $(-1)^F:= -1 \in \Pin^-(d)$.
The map $\Pin^-(d) \to \Z_2$ labeling the two path components of $\Pin^-(d)$ defines $\theta$.
It is given by the composition
\[
\Pin^-(d) \to O(d) \to \Z_2,
\]
where the first map is the double cover and the second map is zero if $\det A = 1$ and one if $\det A = -1$.
Taking $d = 1$ recovers the last example. 
The above discussion generalized straightforwardly to $G = \Pin^+(d)$.
\end{example}

\begin{example}[class AII]
\label{example: class AII}
We lay out one possible physical interpretation of the last example for $d = 2$ as the internal symmetry group of class AII topological insulators.
Pick an orthonormal basis $\{\gamma_1,\gamma_2\}$ of $\R^2$ and also denote by $\gamma_1,\gamma_2 \in \Pin^-(2)$ lifts of reflections along these vectors.
Then there is an isomorphism of topological groups $\Pin^-(2) \cong \frac{U(1)_Q \rtimes \Z_4^{TF}}{\Z_2^F}$, where $\Z_4$ acts on $U(1)$ by complex conjugation after the quotient map to $\Z_2$ and we quotiented out the diagonal $\Z_2$-subgroup.
The isomorphism is given by $T := \gamma_1,$ $iQ := \gamma_1 \gamma_2$.
More precisely, the $U(1)$ subgroup consists of elements of the form $e^{ib Q} = \cos b + \gamma_1 \gamma_2 \sin b$ for $0 \leq b < 2 \pi$.
Physically, we can interpret this symmetry group to consist of time-reversal $T$ of square $(-1)^F$ and charge $Q$ such that the spin-charge relation holds, by which we mean the following.
Focus on $a = \pi$, then since we quotiented out $\Z_2^F$, the group element $e^{i\pi Q}$ becomes fermion parity.
Now consider a particle of charge $n \in \Z$; a vector $v$ inside an irreducible representation $V$ of $U(1)_Q$ with winding number $n$, equivalently $Q v = n v$.
Then $n$ is odd if and only if $v$ is fermionic, i.e. $(-1)^F v = -v$.
\end{example}

Note that fermionic internal symmetry groups in these examples are strictly speaking only spatially internal in the sense that they act trivially on a spatial slice, but not on spacetime.
However, such groups seem to be called internal symmetry groups in the physics literature nonetheless and so we will adopt this language.

\subsection{Same class, different symmetry groups}
\label{subsec:same class}

In this section we address the fact that some symmetry classes in the tenfold way table have different realizations in different frameworks.
The main example in this paper is the case of class D.
In most references, a topological superconductor is described as having `no symmetries' (more precisely, only fermion parity is present)
\cite[section II.A]{altlandzirnbauer}, \cite{kapustinfermionicSPT},\cite{freedhopkins}\cite[section 2.2]{wittenfermionic} \cite{kitaev2009periodic}.
In particular, the Hamiltonians involved are not necessarily charge-conserving.
In other words the internal fermionic symmetry group is $\Z_2^F$.
However, other references (mainly in the mathematical physics community) instead claim that the protecting symmetry group contains both a charge symmetry $Q$ and a particle-hole symmetry $C$ that acts anti-unitarily on one particle Hilbert space \cite{freedmoore} \cite{thiangk-theoretic}.
Such references are often motivated to make the tenfold way table of Schnyder-Ryu-Furusaki-Ludwig \cite[Table 1]{ryutenfoldsurvey} into a theorem.
It might be tempting to interpret class D in this table as assuming an additional particle-hole symmetry on top of the charge symmetry $Q$ in class A.
However, class D is described using BdG Hamiltonians acting on Nambu space in their paper, while class A is described by one-particle Hamiltonians acting on one-particle Hilbert space.
Therefore, classes A and D should not be compared in this way.

To make this discussion concrete, we formulate the subtle difference between these two approaches in the language of Bogliubov-de-Gennes Hamiltonians.
The starting point is a Nambu space $W$ of not-necessarily-charged fields, which is a complex vector space equipped with two structures:
\begin{enumerate}
    \item a nondegenerate symmetric complex-bilinear form $\{.,.\}: W \times W \to \C$ giving the canonical anti-commutation relations;
    \item a complex antilinear involution $\Xi: W \to W$ preserving the nondegenerate bilinear form given by exchanging particles and holes. 
    Following \cite{zirnbauerparticlehole} we will call this \emph{particle-hole conjugation} and remark that some others refer to $\Xi$ as a `particle-hole symmetry', a term that has a different definition in this article, see below.
\end{enumerate}
There is an induced Hilbert space structure on $W$ given by $\langle w_1,w_2 \rangle = \{\Xi(w_1),w_2\}$.
A free BdG Hamiltonian is then a self-adjoint operator $H$ on $W$ which anti-commutes with $\Xi$.
A $G$-symmetry\footnote{For this discussion, we will restrict ourselves to time-preserving symmetries and so $G$ will be a fermionic group with trivial $\theta$.} is a representation $R$ of $G$ on $W$ which commutes with $H$ and $\Xi$ and preserves $\{.,.\}$.
Since fermion parity $(-1)^F \in G$ should act by $-1$ on single fermions we also require it to acts by $-1$ on $W$.

The situation most familiar to condensed matter theorists, is where there is a unit charge\footnote{A representation $R$ of $U(1)_Q$ is called of unit charge if only representations of weight $\pm 1$ occur in the decomposition into irreducibles, i.e. $R(e^{biQ}) = \cos b + R(iQ) \sin b$.} representation of $U(1)_Q$ on $W$ which commutes with $\Xi$ and preserves $\{.,.\}$.
In that case there is a splitting of Nambu spinors $W \cong V \oplus V^*$ into charge $Q = \pm 1$ eigenspaces for a one particle Hilbert space $V$ of electrons.
We can then interpret $W$ as consisting of creation operators $a_v^\dagger$ creating a particle $v \in V$ and annihilation operators $a_{\phi}$ for $\phi \in V^*$.
Note that since $\Xi$ commutes with $e^{ibQ}$ and $\Xi$ is antilinear, $\Xi$ and $Q$ anticommute.
Under the isomorphism $W \cong V \oplus V^*$ we can get that $\Xi$ is induced by the Riesz-Fr\'echet isomorphism $RF: V \cong V^*$ and $\{.,.\}$ is the natural pairing between $V$ and $V^*$.
A BdG Hamiltonian becomes a self-adjoint operator on $W$ of the form
\[
H = 
\begin{pmatrix}
h & \Delta \circ RF^{-1} \\
- RF \circ \Delta & - RF \circ h \circ RF^{-1}
\end{pmatrix},
\]
where we need $h: V \to V$ self-adjoint and $\Delta: V \to V$ to satisfy $\Delta^\dagger = - \Delta$ for $H$ to be self-adjoint. 
Note that the current $U(1)_Q$-representation 
\[
e^{ibQ} = e^{ib} \oplus e^{-ib}: V \oplus V^* \to V \oplus V^*
\]
we are considering on $W$ is a $U(1)_Q$-symmetry if and only if $\Delta = 0$.
This is the charge-conserving case in which $H$ comes directly from a one-particle Hamiltonian $h: V \to V$.
However, in a superconducting system, $H$ does not have to be charge-conserving.
This is the case in which the internal symmetry group $G = \mathbb{Z}^F_2$ consists of nothing but fermion parity.

We will now discuss what happens when a particle-hole symmetry is included in addition to $U(1)_Q$.
We call a symmetry $C \in G$ a \emph{particle-hole symmetry} if it anticommutes with $Q$ \cite[definition (2.1)]{zirnbauerparticlehole}.
Since $C$ exchanges $V$ and $V^*$ we can use $\Xi$ to get a corresponding operator $\Gamma := \Xi C|_V: V \to V$ on one-particle Hilbert space that anticommutes with $h$ and again squares to one.
This is called the pseudosymmetry associated to $C$.
To get a symmetry group equivalent to class D, $\Gamma$ is required to be anti-unitary and square to one.
Equivalently, since $\Xi$ is anti-unitary, we require $C$ to be unitary and $C^2 = 1$.
Since $C$ is $\C$-linear, we get
\[
e^{biQ} C = C e^{-biQ}.
\]
The result is therefore that there is a $U(1)_Q \rtimes \Z_2^C$-symmetry where the action of $\Z_2^C = \langle C: C^2 = 1 \rangle$ on $U(1)_Q$ is by complex conjugation.
Since all group elements involved act unitarily on $W$ and hence on positive energy Fock space, the $\Z_2$-grading $\theta$ is trivial.
Finally, to make $U(1)_Q \rtimes \Z_2^C$ into a fermionic group, we have to choose the central element $(-1)^F$ of square $1$.
We will take it to be $-1 \rtimes 1 \in U(1)_Q \rtimes \Z_2^C$, thereby assuming the spin-charge relation in a similar fashion to example \ref{example: class AII}.
We will briefly discuss what changes in our results if we don't assume the spin-charge relation at the end of section \ref{sec: computing the two}.
In the next section, we will argue that the difference between the two symmetry groups $\Z_2^F$ and $U(1)_Q \rtimes \Z_2^C$ does not matter for the classification in the weakly interacting regime, but it does matter for strongly interacting phases as we will see in section \ref{sec: computing the two}.

By comparing with the case without symmetries above, we can conclude that
\begin{enumerate}
    \item a BdG system with $\Z_2^F$-symmetry consists of a complex Hilbert space $W$ together with an anti-unitary operator $\Xi: W \to W$ such that $\Xi^2 = 1$ and a self-adjoint operator $H: W \to W$ which anticommutes with $\Xi$.
    \item a BdG system with $U(1)_Q \rtimes \Z_2^C$-symmetry consists of a complex Hilbert space $V$ together with an anti-unitary operator $\Gamma: V \to V$ such that $\Gamma^2 = 1$ and a self-adjoint operator $h: V \to V$ which anticommutes with $\Gamma$.
\end{enumerate}
These two pieces of data are in canonical mathematical bijection.
Therefore, from the perspective of free fermions, it is mathematically equivalent to consider systems with $\Z_2^F$-symmetry and systems with $U(1)_Q \rtimes \Z_2^C$-symmetry.
The origin of this is an equivalence between the representation theory of $\Z_2^F$ and the representation theory of $U(1)_Q \rtimes \Z_2^C$ called a Morita equivalence, which we will discuss in more detail in the next section. 

However, this canonical mathematical bijection does not seem to have any physical interpretation, a statement on which we will now elaborate.
First observe that since $\Gamma$ and $\Xi$ are pseudosymmetries, they map positive energy modes to negative energy modes and so are strictly speaking not symmetries.
In more details, suppose $H:W \to W$ is gapped with its Fermi energy shifted to zero.
Then $W = W_+ \oplus W_-$ separates into positive and negative energy modes\footnote{In case there is also charge-conservation, this splitting is typically not related to the splitting $W \cong V \oplus V^*$. Instead $W_+ = V_+ \oplus V_-^*$, where $V_\pm$ are the $\pm$-energy modes of the one-particle Hamiltonian.}.
The BdG Hamiltonian gets second quantized to an operator $\hat{H}: \bigwedge W_+ \to \bigwedge W_+$.
It is free in the sense that it is quadratic in annihilation and creation operators.
If we have a symmetry $g: W \to W$ it second quantizes to an operator $\hat{g}: \bigwedge W_+ \to \bigwedge W_+$ which commutes with $\hat{H}$.
On the other hand, $\Xi$ and $\Gamma$ exchange $W_+$ and $W_-$ and so only define operators $\bigwedge W_{\pm} \to \bigwedge W_{\mp}$.
However, since a pseudosymmetry $\Gamma: V \to V$ has a canonical corresponding symmetry $C$, this symmetry can then be second quantized to a genuine symmetry of the second-quantized Hamiltonian $\hat{H}$ on positive energy Fock space.
Therefore, working with pseudosymmetries is physically justified.
However, applying this procedure in the case where $\Gamma := \Xi$ will lead to a trivial symmetry since $\Xi^2 = 1$.
Therefore unlike the pseudosymmetry $\Gamma$ which naturally corresponds to a genuine particle-hole symmetry $C$, the operator $\Xi$ cannot be interpreted as a symmetry of the second-quantized theory.

To the reader more familiar with the high-energy physics language, the above story can be briefly translated to the language of spinors in high energy physics as follows.
Since we are not necessarily assuming charged particles, the starting point is the real Hilbert space $\mathcal{M}$ of Majorana fields (note that we need Majoranas because we do not assume charge conservation).
We complexify this space to get a complex Hilbert space $W = \mathcal{M} \otimes_\R \C$ with canonical complex conjugation $\Xi$ and a symmetric nondegenerate complex-bilinear form
\[
\{w_1,w_2\} = \langle \Xi(w_1),w_2 \rangle.
\]

To end this section, we briefly also discuss the example of class AI.
According to some references such as \cite{freedmoore} and \cite{freedhopkins}, the internal symmetry group is $U(1)_Q \rtimes \Z_2^T$ in which now the generator $T \in \Z_2^T$ is anti-unitary.
One other way to realize class AI that seems natural from a physical perspective was already given in the original paper of Altland-Zirnbauer \cite{altlandzirnbauer}: it includes time-reversal, charge and spin.
We then arrive instead at a fermionic symmetry group which roughly
\footnote{The exact form of $G$ depends on `how discrete' we want to pick it. For example, instead of $SU(2)$ we could take the finite quaternion group.} 
looks like
\[
G = \frac{SU(2) \times U(1)_Q \rtimes \Z_4^{FT}}{\Z_2^F \times \Z_2^F},
\]
where the quotient is to ensure that the fermion parities of all three factors are identified.
In particular we chose to assume both the spin-charge and spin-statistics relations in this symmetry group.
Independent of the exact form of $G$, it is not isomorphic to the group $U(1)_Q \rtimes \Z_2^T$.

We argue in the next section that two different choices of symmetry groups can be mathematically equally good to describe some symmetry class X if their representation theory is sufficiently similar for the $K$-theory classifications to agree.
However, this reasoning breaks down for interacting phases, since these no longer depend only on the $K$-theory.
It would be interesting to work out the difference between these two approaches to class AI in the interacting setting similar to how we handle class D in section \ref{sec: computing the two}.

\subsection{Morita equivalence of internal symmetries}
\label{subsec:Morita equiv}

The two internal symmetry groups $\Z_2^F$ and $U(1)_Q \rtimes \Z_2^C$ of class D that have been proposed in the last section are different, but they are similar in the sense that they have equivalent representation theory in a suitable sense.
Mathematicians call this Morita equivalence and this turns out to be good enough to have the same classification of SPT-phases in the weakly interacting limit.
This is because $K$-theory is Morita-invariant.

More precisely, let us focus on the real representation theory of these internal symmetry groups needed to describe individual Majorana fermions with this symmetry living in some real Hilbert space $\mathcal{M}$.
Since this space describes single fermions, $(-1)^F$ acts as $-1$.
So the symmetry algebra
\[
\frac{\R[(-1)^F]}{((-1)^F = -1)} \cong \R
\]
is trivial and being a representation of $\R$ requires no extra data on $\mathcal{M}$.

Now assume the Majorana additionally furnishes a representation of $U(1)_Q \rtimes \Z_2^C$ and has charge $\pm 1$ under $Q$.
Then the symmetry algebra is instead a $2\times 2$-matrix algebra
\[
\frac{\R[iQ, C]}{((iQ)^2 = -1, C^2 = 1, iQC = -CiQ) } \cong Cl_{1,1} \cong M_2(\R),
\]
so the data of this Majorana is simply a representation of $M_2(\R)$.
In particular, the charge and particle-hole symmetry require the real Hilbert space to be even-dimensional.

The point of this discussion is that the algebras $\R$ and $M_2(\R)$ are not isomorphic.
However, they are Morita equivalent because their categories of representations are equivalent: the correspondence is given by mapping a Hilbert space to a Hilbert space of twice its size.
This means that $K$-theory cannot distinguish these two algebras.
Since SPT-phases of free fermions are described by $K$-theory, the symmetry groups $\Z_2^F$ and $U(1)_Q \rtimes \Z_2^C$ will result in the same class D classification after assuming unit charge.
For concreteness, this classification\footnote{We ignore weak topological invariants for simplicity.} is given by the Bott-periodic expression
\[
KO^{d-3}(\pt),
\]
where $d$ is the spacetime dimension, also see table \ref{tab:ten-fold}.

However, we emphasize that whether this implies that the two classifications are `equal' depends on how you compare them; a natural comparison map is in the direction $SPT^d_{U(1)_Q \rtimes \Z_2^C} \to SPT^d_{\Z_2^F}$ given by breaking the symmetry.
On the algebraic level in the example above, this corresponds to mapping a certain type of representation of $M_2(\R)$ to its underlying representation of $\R$.
Since there exist real vector spaces of odd dimension, this mapping is typically not surjective.
Instead it is a kind of multiplication by two map.
In particular, even though the individual $K$-theory groups are isomorphic, this specific comparison map is not an isomorphism.


\section{Interacting SPT-phases Using Bordism}
\label{sec: interacting}

In contrast to the free case, SPT phases with strong interactions are mathematically on a much less robust footing.
In recent years, much progress has been made from many angles.
Here we review the approach of \cite{freedhopkins} (also see \cite{kapustinfermionicSPT}) to study SPT-phases using invertible field theories.
We provide a new approach to their ansatz for a general fermionic internal symmetry group in section \ref{sec: spacetime structure groups}.
In section \ref{sec: computing the two} we calculate the group of SPT-phases for the two class D internal symmetry groups $\Z_2^F$ and $U(1)_Q \rtimes \Z_2^C$.

\subsection{Low-energy effective theories of invertible lattice models}
\label{sec: low-energy effective}

The core idea is to study condensed matter systems by their low-energy effective continuum quantum field theory, which conjecturally is sufficiently simple to be defined mathematically.
The proposal is that there is a map \cite[section 9.1]{freedlectures}
\[
\{ \text{invertible gapped lattice systems with symmetry }G \} \to ITFT^d_G
\]
from a still-to-be-defined space of condensed matter systems to a space of certain\footnote{These are not topological field theories in the classical sense of Atiyah-Segal in which partition functions take values in $\C^\times$ with the discrete topology, but so-called `topological* field theories' in which we give $\C^\times$ the continuous topology and the theory can be mildly non-topological.} functorially defined quantum field theories that are invertible under stacking. 
One notable case in which this map has been worked out rigorously is for the time-reversal symmetric Kitaev chain \cite{debraygunningham}.
Conjecturally \cite{freedlectures} the map is an isomorphism on path components
\begin{align*}
\{G\text{-SPT}& \text{ phases}\} 
\\
&    \rotatebox[origin=c]{90}{$=$}
\\
\pi_0(\{ \text{invertible gapped lattic}&\text{e systems with symmetry }G \}) 
\\
&    \rotatebox[origin=c]{-90}{$\longrightarrow$}
\\
 \pi_0(IT&FT^d_G) 
 \\
 &\rotatebox[origin=c]{90}{$=$}
 \\
 \{\text{deformation classes of unitary inv}&\text{ertible }G\text{-topological* field theories}\}
\end{align*}
Therefore up to the issue of defining the relevant notions, we can indirectly study symmetry-protected topological phases by studying invertible (roughly) topological field theories.\footnote{The invertibility assumption is essential here: fractons provide an example of lattice systems of which the low energy effective theory is not topological*.}
This turns out to be useful, since topological field theories can be defined rigorously and the homotopy type of invertible field theories is very well-understood \cite{chrisinvertible}.
Three nontrivial problems addressed in \cite{freedhopkins} are
\begin{enumerate}
    \item to define the appropriate `target spectrum' for the field theory;
    \item to define when an invertible topological field theory is unitary;
    \item given an internal fermionic symmetry group $G$, they define Euclidean $d$-dimensional spacetime structure groups $H_d(G)$ which they use to equip their bordisms with the appropriate structure (such as spin).
\end{enumerate}
Since Freed-Hopkins mainly focus on their ten-fold way, we will outline a procedure to arrive at the third point above for a general fermionic internal symmetry group in section \ref{sec: spacetime structure groups}.
We will not go into detail on the other points, but instead be satisfied with the fact that their proposed deformation classes of invertible unitary topological field theories in spacetime dimension $d$ is (noncanonically) isomorphic to 
\[
\operatorname{Torsion}\left(\Omega^{H(G)}_d\right) \oplus \operatorname{Free}\left(\Omega^{H(G)}_{d+1}\right)
\]
for compact symmetry groups $G$.
Here $\Omega^{H(G)}_d$ denotes the bordism group of $d$-dimensional manifolds with $H(G)$-structure and $\operatorname{Torsion}$ and $\operatorname{Free}$ denote the torsion respectively free part of the group.

\subsection{Spacetime structure groups}
\label{sec: spacetime structure groups}

To define the correct structure groups on Euclidean spacetime, we have to face some modifications of our symmetry group by Wick rotation.
For example, in usual Lorentzian signature a time-reversal symmetry will satisfy $T^2 = (-1)^F$, but after Wick rotation this sign will change to $+1$.
Therefore the spacetime structure group on a $d$-dimensional topological field theory of a fermionic system with time-reversal is $\Pin^+_d$, not $\Pin^-_d$ as one might naively expect.
We approach this problem using the following construction, which we think of as a tensor product over $\Z_2^F$ (also see the preprint \cite{stolzpreprint}, and the proof of \cite[theorem 2.2.1]{teichnerthesis}).
Recall that we write the $\Z_2$-grading $\theta(g) \in \{0,1\}$ additively.

\begin{definition}
Let $(G, \theta_G, (-1)^F_G), (H,\theta_H,(-1)^F_H)$ be fermionic groups.
Then the \emph{fermionic tensor product} $G \otimes H $ is equal to $(G \times H)/\langle((-1)^F_G, (-1)^F_H )\rangle$ as a space.
We endow it with a product
\[
(g_1 \otimes h_1) (g_2 \otimes h_2) = ((-1)^F_G)^{\theta_G(g_2)\theta_H(h_1)} g_1 g_2 \otimes h_1 h_2
\]
the central element $1 \otimes (-1)^F_H = (-1)^F_G \otimes 1 \in G \otimes H$ and the grading $\theta(g \otimes h) = \theta_G(g) + \theta_H(h)$.
When it should be clear from the context we omit the subscripts $G$ and $H$ from $\theta$ and $(-1)^F$ to improve readability.
\end{definition}

The proof of the following proposition is a direct computation which is included for completeness.

\begin{proposition}
The fermionic tensor product $G \otimes H$ is a fermionic group.
\end{proposition}
\begin{proof}
To show that the operation is well-defined, consider $g_1' = (-1)^F g_1 \in G$ and $h_1' = (-1)^F h_1$. 
Then because $\theta_H((-1)^F_H) = 0$,
\begin{align*}
((-1)^F)^{\theta(g_2)\theta(h_1')} &g_1' g_2 \otimes h_1' h_2 
\\
&=  ((-1)^F)^{\theta(g_2)(\theta(h_1) + \theta((-1)^F))} (-1)^F g_1 g_2 \otimes (-1)^F h_1 h_2 
\\
&= ((-1)^F)^{\theta(g_2)\theta(h_1)} g_1 g_2 \otimes h_1 h_2.
\end{align*}
There is a similar computation for changing the second argument of the product, which shows it is independent of chosen representatives.

To show that $G \otimes H$ is a fermionic group, note first that $1 \otimes 1$ is a unit and 
\[
(g \otimes h)^{-1} = ((-1)^F)^{\theta(g) \theta(h)} g^{-1} \otimes h^{-1}.
\]
It is also clear that $(-1)^F \otimes 1$ is central of degree zero and the degree $\theta: G \otimes H \to \Z/2$ is a homomorphism.
To show associativity, we pick $g_1,g_2,g_3 \in G$ and $h_1,h_2,h_3 \in H$ and compute
\begin{align*}
    (g_1 \otimes h_1)((g_2 \otimes h_2)&(g_3 \otimes h_3)) 
    \\
    &=(g_1 \otimes h_1)((-1)^F)^{\theta(h_2)\theta(g_3)} g_2 g_3 \otimes h_2 h_3)
    \\
    &= ((-1)^F)^{\theta(h_1)\theta(g_2g_3)} ((-1)^F)^{\theta(h_2)\theta(g_3)} g_1 g_2 g_3 \otimes h_1 h_2 h_3,
\end{align*}
where in the last line the fact that $\theta(((-1)^F)^{\theta(h_2)\theta(g_3)}) = 0$ was used.
A similar computation gives
\begin{align*}
    ((g_1 \otimes h_1)(g_2 \otimes h_2))&(g_3 \otimes h_3) 
    \\
    &= ((-1)^F)^{\theta(h_1)\theta(g_2)} ((-1)^F)^{\theta(g_3)\theta(h_1 h_2)} g_1 g_2 g_3 \otimes h_1 h_2 h_3.
\end{align*}
The proof is finished by comparing the exponents of the central element:
\begin{align*}
    \theta(h_1)\theta(g_2) + \theta(g_3) \theta(h_1h_2) &= \theta(h_1)\theta(g_2) + \theta(g_3) \theta(h_1) + \theta(g_3)\theta(h_2)
    \\
    &= \theta(h_1)(\theta(g_2) + \theta(g_3)) + \theta(g_3)\theta(h_2)
    \\
    &= \theta(h_1)\theta(g_2g_3) + \theta(g_3)\theta(h_2).
\end{align*}
\end{proof}

Note that $G \otimes H$ is actually naturally bigraded; there is a homomorphism $\theta_G \otimes \theta_H: G \otimes H \to \Z_2 \times \Z_2$.
The grading that we used to make $G \otimes H$ into a fermionic group is the bigrading composed with the group operation $\Z_2 \times \Z_2 \to \Z_2$.
We often consider the even part $(G \otimes H)_{ev}$ under this grading.
Note that the other part of the bigrading still gives a grading $(G \otimes H)_{ev}  = G_{ev} \otimes H_{ev} \sqcup G_{odd} \otimes H_{odd}$.
As such $(G \otimes H)_{ev}$ becomes a fermionic group.

\begin{definition}
Let $(G, (-1)^F, \theta)$ be a compact fermionic symmetry group, which we think of as the internal symmetry group.
The associated \emph{$d$-dimensional space-time structure group} is
\[
H_d(G) := (\Pin^+(d) \otimes G)_{ev}.
\]
The orthogonal representation $\rho: H_d(G) \to O_d$ that we will use as a structure map is induced by the projection onto the first factor.
\end{definition}

\begin{example}
Let $G$ be bosonic without time-reversing symmetries, so $(-1)^F = 1$ and $\theta = 0$.
Then the spacetime structure group is
\[
H_d(G) = SO(d) \times G.
\]
\end{example}

\begin{example}
\label{ex:SU(2)}
If $G$ is a fermionic group without time-reversing symmetries, then
\[
H_d(G) = \frac{\Spin(d) \times G}{\Z_2^F},
\]
where we quotiented by the diagonal $\Z_2^F$.
For example, let $G = SU(2)$ with $(-1)^F := -\operatorname{id}$.
The resulting spacetime structure group 
\[
H_d(SU(2)) = \frac{\Spin(d) \times SU(2)}{\Z_2^F}
\]
has received some attention in the recent mathematical physics literature.
It is the main topic of \cite{wangnewsSU(2)} where it is called $\Spin_{SU(2)}(d)$, it is called $G^0_d$ in \cite{freedhopkins} and perhaps the most suitable name is $\Spin^h(d)$ introduced in \cite{chen2017bundles}.

The $SU(2) = \Spin(3)$-group occurring here can be interpreted as internal spin degrees of freedom of the particle.
 This allows us to relate the example to the original applications to topological condensed matter. 
 Namely, it appears as internal symmetry of the class C spin singlet superconductor, which is favoured in materials with negligible spin-orbit interaction.
Note that this physical interpretation is $3+1$d-centric from a particle physics perspective: we really need our electrons to be massive $3+1$d spinors to get internal $\Spin(3)$-degrees of freedom.
If we want to describe SPT-phases in class C in say dimension $2+1$, we imagine our electrons to be confined in a spatial plane, but still have the internal $\Spin(3)$-symmetry.
\end{example}

\subsection{The group of SPT-phases}

Recall that a superalgebra is a $\Z_2$-graded algebra in which the multiplication respects the $\Z_2$-grading.
For example, the Clifford algebra with (say) positive squares
\[
Cl_{+p} = \frac{\R[\gamma_1, \dots, \gamma_p]}{\gamma_i \gamma_j + \gamma_j \gamma_i = 2 \delta_{ij} }
\]
is a real superalgebra in which all $\gamma_i$ are odd.
The simplest superalgebras are the superdivision algebras: superalgebras in which every homogeneous element is invertible.
Algebraically, the origin of the tenfold way arises from the following super-version of the Frobenius theorem on real division algebras \cite{geikomoore}.

\begin{theorem}[\cite{wall1964graded}]
There are ten superdivision algebras over $\R$:
\begin{align*}
&\C, \quad \C l_1, 
\\
&\R, \quad Cl_{-1}, \quad Cl_{-2}, \quad Cl_{-3}, \quad \mathbb{H}, \quad Cl_{3}, \quad Cl_{2}, \quad Cl_{1}
\end{align*}
\end{theorem}

There is a canonical way to pass from superdivision algebras to fermionic groups, giving us a special family of ten fermionic internal symmetry groups as follows.
If $D$ is a super division algebra, define the sphere of $D$ to be the quotient group
\[
S(D) := \frac{D^\times}{\R_{>0}},
\]
where $D^\times \subseteq D$ is the topological group of homogeneous (hence invertible) elements and $\R_{>0} \subseteq D^\times$ the subgroup of positive scalars.
We make $S(D)$ into a fermionic group by using $(-1)^F = [-1] \in S(D)$ and $\theta[d] = |d|$ is the supergrading of $D$.
The list of ten fermionic groups we obtain is the list of non-relativistic internal symmetry groups of the ten-fold way in the formulation of \cite[(9.34),(9.35)]{freedhopkins}
\begin{align*}
&\Spin^c(1), \Pin^c(1), 
\\
&\Spin(1), \Pin^-(1), \Pin^-(2), \Pin^-(3), 
\Spin(3), \Pin^+(3), \Pin^+(2), \Pin^+(1).
\end{align*}
In particular, the case $\Pin^-(1)$ recovers example \ref{example: class DIII} and the case $\Pin^-(2)$ recovers example \ref{example: class AII}.
The fermionic groups 
\[
H_d(S(D)) = (\operatorname{Pin}^+(d) \otimes S(D))_{ev}
\]
will give the ten spacetime structure groups of \cite{freedhopkins}\footnote{There is a sign difference between \cite{freedhopkins} and this formulation. For example, in \cite{freedhopkins}, the non-relativistic internal symmetry group $I = \Pin^+(1)$ corresponds to the structure group $H_d = \Pin^+(d)$, while in the formulation of this paper it corresponds to $\Pin^-(d)$. The convention chosen here seems to be in agreement with the physics literature.}, see table \ref{tab:ten-fold}.
For example, $D = \R$ gives $S(D) = \Spin_1 = \Z_2^F$ and $H_d(S(D)) = \Spin(d)$.
Example \ref{ex:SU(2)} corresponds to the case $D = \mathbb{H}$.

Next the idea is to endow spacetime with a $H_d(G)$-structure, see \cite[section 2.2]{freedhopkins} for the definition of an $H$-structure on a manifold for a group $H$ together with a map to $O_d$.
We then consider unitary invertible field theories with this structure group, which can be described using bordism groups.
Namely, the proposal of \cite[remark 8.39]{freedhopkins} extended to the current setting is

\begin{ansatz}
Given a fermionic internal symmetry group $G$, the abelian group of deformation classes of $d$-dimensional $G$-protected SPT-phases is
\[
SPT^d_G := [MTH(G), \Sigma^{d+1} I \Z].
\]
\end{ansatz}

Here $MTH(G)$ is the Madsen-Tillmann spectrum \cite{gmtw}, $I \Z$ the Anderson-dual of the sphere spectrum, and the square brackets denote homotopy classes of maps of spectra.
For more details, we refer the reader to \cite[section 5.3 and 7.1]{freedhopkins}. 
However, in this paper we will only consider the right hand side of the noncanonical isomorphism
\[
SPT^d_G \cong \operatorname{Torsion}\left(\Omega^{H(G)}_d\right) \oplus \operatorname{Free}\left(\Omega^{H(G)}_{d+1}\right)
\]
which holds for every compact $G$.
Here $\Omega^{H(G)}_d$ denotes the bordism group consisting of $d$-dimensional $H_d(G)$-manifolds modulo the ones that bound a $d+1$-dimensional $H_{d+1}(G)$-manifold.

\begin{example}
Consider the class D internal symmetry group $\Z_2^F$.
The bordism groups corresponding to the spacetime structure group $H_d(\Z_2^F) = \Spin(d)$ are well-known \cite{anderson1967structure}.
Using this we can find the group of SPT-phases. 
For example, in dimension $2+1$ 
\[
SPT^3_{\Z_2^F} \cong \operatorname{Torsion}(\Omega^{\Spin}_3) \oplus \operatorname{Free}(\Omega^{\Spin}_4) = \Z
\]
is the group of SPT-phases in dimension $2+1$ protected by fermion parity only.
The partition function of the generator is given by the Atiyah-Patodi-Singer eta invariant of the Dirac operator.
\end{example}

In the next section we will show that $SPT_{U(1)_Q \rtimes \Z_2^C}^3$ is not isomorphic to $SPT_{\Z_2^F}^3$ by computing the relevant bordism groups.
This will show that in the approach to SPT-phases using bordism, these two symmetry groups describing class D give a different classification of topological phases.

\begin{table}[h!] 
\centering
{
\renewcommand{\arraystretch}{1.5}
\resizebox{\columnwidth}{!}{
\begin{tabular}{l|l|l||l|l|l|l|l|l|l|l}
Symmetry class & A & AIII & D & BDI & AI & CI &  C & CII & AII & DIII 
\\
\hline
Superdivision algebra $D$ & $\C$ & $\C l_1$ & $\R$ & $Cl_{+1}$ & $Cl_{+2}$& $Cl_{+3}$ & $\mathbb{H}$ & $Cl_{-3}$& $Cl_{-2}$& $Cl_{-1}$
\\
\hline 
Internal symmetry $S(D)$ & $Spin_1^c$ & $Pin_1^c$ & $Spin_1$ & $Pin_1^+$  & $Pin_2^+$ & $Pin_3^+$ & $Spin_3$ & $Pin_3^-$ & $Pin_2^-$ & $Pin_1^-$ 
\\
\hline
$K$-theory group & $K^{d+1}$ & $K^d$ & $KO^{d-3}$ & $KO^{d-2}$ & $KO^{d-1}$ & $KO^{d}$ & $KO^{d+1}$ & $KO^{d+2}$ & $KO^{d+3}$ & $KO^{d+4}$ 
\\
\hline
Spacetime structure group $H$ & $Spin^c_d$ & $Pin^c_d$ & $Spin_d$ & $Pin^-_d$ & $Pin^{\tilde{c}-}_d$ & $G^+_d$ & $G^0_d$ & $G^-_d$ & $Pin^{\tilde{c}+}_d$ & $Pin^+_d$ 
\end{tabular}
}
}
\caption{The tenfold way table in the formulation of this paper.
The internal symmetry group $S(D)$ is chosen to correspond to one of the ten real superdivision algebras as defined above.
The $K$-theory group classifies the free fermion SPT-phases of spacetime dimension $d$ in the corresponding symmetry class (ignoring weak invariants) and we implicitly evaluated the $KO$-theory on a point.
The spacetime structure group is the combination of the internal symmetry and the (Euclidean) Lorentz symmetry endowing spacetime with an $H$-structure discussed in section \ref{sec: spacetime structure groups}.
See \cite{freedhopkins} for the notation used for $H$.
}\label{tab:ten-fold}
\end{table}

\section{Class D Classification for Interacting Theories}

\subsection{Computing the other class D bordism group}
\label{sec: computing the two}

In this section we will compute the group of $H$-manifolds up to bordism in low dimensions, where (see example \ref{ex:SU(2)})
\[
H_d := H_d(G) = \frac{\Spin(d) \times U(1)_Q \rtimes \Z_2^C}{\Z_2^F}
\]
is the spacetime structure group associated to the fermionic internal symmetry group $G = U(1)_Q \rtimes \Z_2^C$ consisting of charge $Q$ and charge-conjugation $C$.
To this extent we will employ an Atiyah-Hirzebruch type spectral sequence originally introduced in \cite[theorem 3.1.1]{teichnerthesis} under the name `James spectral sequence' as follows.
Consider the short exact sequence of topological groups (also see proposition \ref{prop: exact} below)
\begin{equation}
\label{eq: ses}
1 \to \Spin(d) \to H_d \to U(1) \rtimes \Z_2 \to 1,    
\end{equation}
where the map $H_d \to U(1) \rtimes \Z_2$ is given by $[r, z, y] \mapsto (z^2,y)$, where $r \in \Spin_d, z \in U(1)$ and $y \in \Z_2$.
Note that $U(1) \rtimes \Z_2 \cong O(2)$ by mapping the $\Z_2$ to any reflection in $\R^2$.
This short exact sequence is compatible with the inclusion maps $H_{d} \to H_{d+1}$ and $\Spin(d) \to \Spin(d+1)$.
So taking the colimit as $d \to \infty$, there is an associated fibration of topological spaces
\[
B \Spin \to B H \to BO(2)
\]
in which the first two spaces have compatible structure maps to $BO$.
We will use a generalized version of the Serre spectral sequence associated to this fibration which looks like
\[
H_p(BO(2), \Omega^{\Spin}_q) \implies \Omega^H_{p+q}.
\]
Just as in the ordinary Serre spectral sequence, we might have to deal with local coefficients.
In our case the base is not simply connected: $\pi_1(BO(2)) = \pi_0(O(2)) = \Z_2$.
However, the associated homotopy action of $\Z_2$ on the fiber of the fibration is trivial.
This can be seen from the short exact sequence (\ref{eq: ses}): the element $1 \otimes C \in H_d$ maps to the nontrivial connected component of $O(2)$ but commutes with the subgroup $\Spin_d$.
Therefore the coefficients are not twisted.
We can now look up the first few spin bordism groups \cite{kirbytaylor}
\begin{equation}
\label{eq:spin bordism}
\Omega^{\Spin}_0 = \Z, \quad \Omega^{\Spin}_1 = \Omega^{\Spin}_2= \Z_2, \quad \Omega^{\Spin}_3 = 0, \quad \Omega^{\Spin}_4 =\Z.
\end{equation}
Writing the base space as $B = BO(2)$ to save space, the relevant part of the second page of the spectral sequence looks like

\begin{center}
\begin{tikzcd}[row sep={2.5em,between origins}, column sep={6em,between origins}]
H_0(B,\Z) & \dots &&&&\\
0               &  0                &   0               & \dots  && \\
H_0(B,\Z_2)    & H_1(B,\Z_2)      & H_2(B,\Z_2)      & H_3(B,\Z_2)      & \dots &   \\
H_0(B,\Z_2)    & H_1(B,\Z_2)      & H_2(B,\Z_2)      & H_3(B,\Z_2)      & H_4(B,\Z_2)    & \dots      \\
H_0(B,\Z) & H_1(B,\Z)   & H_2(B,\Z)   & H_3(B,\Z)   & H_4(B,\Z)   & H_5(B,\Z)   
\end{tikzcd}
\end{center}

The cohomology of $BO(2)$ with $\Z_2$-coefficients is well-known: \cite{milnorstasheff}
\[
H^*(BO(2),\Z_2) = \Z_2[w_1,w_2],
\]
where $w_1$ is the first Stiefel-Whitney class in degree one and $w_2$ the second Stiefel-Whitney class in degree two.
Since $\Z_2$ is a field, the homology is in duality with the cohomology:
\[
H_*(BO(2),\Z_2) = \Hom (H^*(BO(2),\Z_2),\Z_2)= \Z_2[\nw_1, \nw_2].
\]
so the homology is spanned by the dual basis of the basis consisting of $w_1^k w_2^l$, which we write $\nw_1^k \nw_2^l$.
We emphasize that $\nw_1 \nw_2$ is by no means the product of $\nw_1$ and $\nw_2$.
The integral homology of $BO(n)$ for general $n$ was computed in \cite{feshbach1983integral}.
Alternatively one can compute it using Bockstein homology in a similar way to \cite[section 3.E]{hatcher} using the fact that 
\begin{align*}
    H^*(BO(2),\Z[1/2]) &= \Z[1/2,p_1],
\end{align*}
see \cite[theorem 15.9, problem 15B]{milnorstasheff}.
The result in the dimensions of interest is given in table \ref{tab: cohomology of BO}.
Plugging in these results yields the second page in figure \ref{second page}.

\begin{sidewaysfigure}
    \centering
    \label{second page}
\begin{tikzcd}[row sep={6em,between origins},column sep={8em,between origins},ampersand replacement=\&]
\Z \& \dots \&\&\&\&\\
0       \&  0                            \&   0               \& \dots  \&\& \& \\
\Z_2    \& \Z_2\langle\nw_1 \rangle      \& \Z_2\langle\nw_1^2 , \nw_2 \rangle   \& \Z_2\langle\nw_1 \nw_2 , \nw_1^3 \rangle      \& \dots \&\&   \\
\Z_2    \& \Z_2\langle\nw_1 \rangle      \& \Z_2\langle\nw_1^2, \nw_2 \rangle \arrow[llu,"\sim",swap]  \& \Z_2\langle\nw_1 \nw_2 , \nw_1^3 \rangle \arrow[llu,two heads]\& \Z_2\langle\nw_2^2,\nw_2\nw_1^2, \nw_1^4\rangle \arrow[llu]    \& \dots \&     \\
\Z      \& \Z_2\langle\nw_1 \rangle      \& \Z_2\langle \nw_2 \rangle       \arrow[llu,"\sim", swap]    \&  \Z_2\langle\nw_1^3 \rangle     \arrow[llu, "0",swap] \& \Z \langle \np_1 \rangle \oplus \Z_2 \langle \nw_1^2 \nw_2 \rangle \arrow[llu]  
\& \Z_2\langle \nw_1 \nw_2^2, \nw_1^5 \rangle \arrow[llu]\& \dots \\
\&\&\&\&\&
\end{tikzcd}
\caption{The second page of the spectral sequence to compute $\Omega^H_d$. 
Several second differentials are displayed.}
\end{sidewaysfigure}
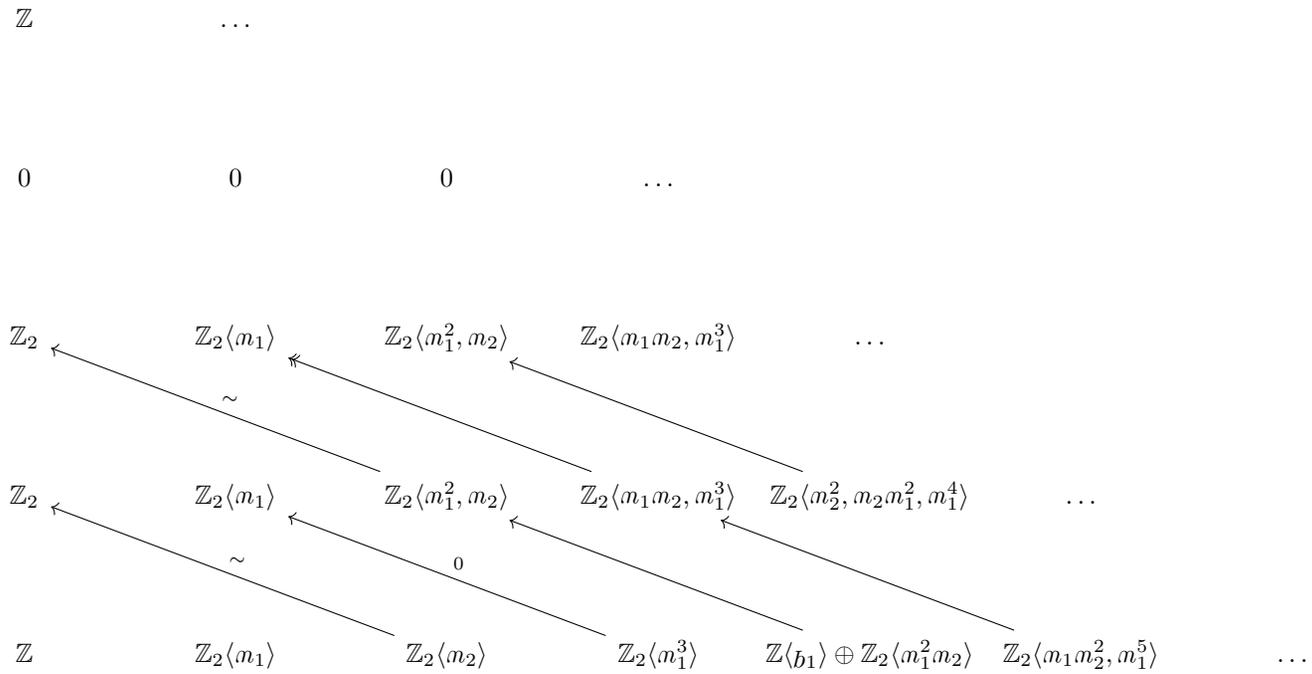

\if
But since this integral homology is less well-known we will compute $H_*(BO(2),\Z)$ in low dimensions using elementary considerations.
From \cite[theorem 15.9, problem 15B]{milnorstasheff} we retrieve:
\begin{align*}
    H^*(BO(2),\Z[1/2]) &= Z[1/2,p_1],
\end{align*}
where $p_1$ is the first Pontryagin class in degree $4$.
To compute the integral cohomology of this, we make use of a form of Bockstein cohomology (see section 3.E of \cite{hatcher}).
For this we have to compute how the first Steenrod square $Sq^1$ acts on the cohomology.
Because $w_1$ is of degree one, $Sq^1(w_1) = w_1^2$.
It follows from problem 8A in \cite{milnorstasheff} that 
\[
Sq^1 w_2 = w_1 w_2,
\]
since $w_3 = 0$ in $BO(2)$.
Using that $Sq^1$ is a derivation, we can now compute it on the full cohomology ring. 
Now consider the commutative diagram

\[
\begin{tikzcd}
\dots \arrow[r]& H^n(BO(2),\Z) \arrow[r,"\pmod 2"] \arrow[d, "\pmod 4"] & H^n(BO(2),\Z_2) \arrow[r] \arrow[d,equal] & H^{n+1}(BO(2),\Z) \arrow[d, "\pmod 2"] \arrow[r] & \dots\\
\dots \arrow[r] & H^n(BO(2),\Z_4) \arrow[r,"\pmod 2"] & H^n(BO(2),\Z_2) \arrow[r,"Sq^1"] & H^{n+1}(BO(2),\Z_2) \arrow[r]& \dots
\end{tikzcd}
\]

Note that if $x \in H^n(BO(2),\Z_2)$ lifts to integral cohomology, then $Sq^1 x = 0$.
Also, if $Sq^1 y \in H^n(BO(2),\Z_2)$ is in the image of the Steenrod square, then $Sq^1 y$ lifts to an integral cohomology class that is $2$-torsion.
So for example, we know that $Sq^1 w_1 = w_1^2$ and so $w_1^2$ lifts to a $2$-torsion class in $H^2(BO(2),\Z)$ and $w_1$ has no integral lift.
Unfortunately, if we only know that $x \in H^n(BO(2),\Z_2)$ satisfies $Sq^1 x = 0$, but not whether it is in the image, then we can only conclude that $x$ lifts to a $\Z_4$-class.
However, this is where our knowledge of the cohomology with $\Z[1/2]$ might help us out to see whether it lifts further to an integral class (alternatively, we could use results in \cite[section 30.5]{borel1959characteristic} that the torsion in the cohomology $BO(n)$ is all $2$-torsion at this point).
As an example, note that $Sq^1 (w_2^2) = 0$, but $w_2^2$ is not in the image.
Therefore it lifts to an order $4$ class in $H^4(BO(2),\Z_4)$.
We see that $w_2^2$ has an integral lift $p_1$, which maps to the multiplicative generator of $H^*(BO(2),\Z[1/2])$.

\textcolor{red}{Explanation to get to homology}
\fi

\begin{table}
    \centering
    \begin{tabular}{c|c|c|c|c|c|c}
        $i$             & $0$ & $1$ & $2$           & $3$           & $4$                           & $5$                  \\ \hline
        $H_i(BO(2),\Z_2)$ & $\Z_2$& $\Z_2$ & $\Z_2^2$        & $\Z_2^2$        & $\Z^2_2$              & $\Z_2^3$             \\ \hline
        generators      & $1$ & $\nw_1$ & $ \nw_1^2, \nw_2$   & $\nw_1 \nw_2, \nw_1^3$   & $\nw_2^2,\nw_2\nw_1^2, \nw_1^4$           & $\nw_1 \nw_2^2, \nw_1^3 \nw_2^2, \nw_1^5$  \\ \hline
        $H_i(BO(2)   ,\Z)$ &$\Z$& $\Z_2$ & $\Z_2$         & $\Z_2$        & $\Z \oplus \Z_2$            & $\Z_2^2$                  \\ \hline 
        generators      & $1$ & $\nw_1$ & $\nw_2$   & $ \nw_1^3$   & $\np_1,   \nw_1^2 \nw_2   $ & $\nw_1 \nw_2^2$, $\nw_1^5$
    \end{tabular}
    \caption{Low-dimensional homology groups of $BO(2)$. Here $\protect \nw_1^k \protect \nw_2^l$ is the dual basis of the basis $w_1^k w_2^l$ of $H^*(BO(2),\Z_2)$. If they exist, we denote their integral lifts to two-torsion classes by the same symbol. Furthermore $\protect\np_1$ is an integral lift of $\protect\nw_2^2$ generating the $\Z$-part.}
    \label{tab: cohomology of BO}
\end{table}

To get to the third page, there are general expressions for the second differentials (theorem 3.1.3 in \cite{teichnerthesis}).
In the case at hand, the differentials of the form $d_2: H_{i+2}(BO(2),\Z_2) \to H_{i}(BO(2),\Z_2)$ are dual to 
\[
Sq^2 + w_2\cdot: H^{i}(BO(2),\Z_2) \to H^{i+2}(BO(2),\Z_2).
\]
The differentials of the form $d_2: H_{i+2}(BO(2),\Z) \to H_{i}(BO(2),\Z_2)$ are the same, but only after reduction modulo two $H_{i+2}(BO(2),\Z) \to H_{i+2}(BO(2),\Z_2)$.

The map $\delta := Sq^2 + w_2$ is given as follows in degrees at most five:
\begin{align*}
    &\delta(1) = w_2, \quad \delta w_1 = w_1 w_2, \quad \delta w_2 = 0, \quad \delta(w_1^2) = w_1^4 + w_1^2 w_2, 
    \\
    &\delta(w_1^3) = w_1^5 + w_1^3 w_2, \quad \delta(w_1 w_2) = w_1^3 w_2.
\end{align*}
To compute these, one can use the Cartan formula, the fact that $Sq^i x = x^2$ if $|x| = i$ and $Sq^i x = 0$ if $|x| < i$.
One other useful equation is
\[
Sq^1 w_2 = w_1 w_2,
\]
which follows for example from problem 8A in \cite{milnorstasheff} since $w_3 = 0$ in $BO(2)$.
The dual $\delta^*: H_*(BO(2),\Z_2) \to H_{*-2}(BO(2),\Z_2)$ of $\delta$ in these degrees is
\begin{align*}
    \delta^* \nw_2 &= 1 \quad \delta^* \nw_1^2 = 0 \quad \delta^* \nw_1^3 = 0 \quad \delta^* \nw_1 \nw_2 = \nw_1
    \\
    \delta^* \nw_1^4 &= \nw_1^2 \quad \delta^* \nw_1^2 \nw_2 = \nw_1^2 \quad \delta^* \nw_2^2 = 0
    \\
    \delta^* \nw_1^5 &= \nw_1^3 \quad \delta^* \nw_1^3 \nw_2 = \nw_1^3 + \nw_1 \nw_2 \quad \delta^* \nw_1 \nw_2^2 = 0.
\end{align*}
Using these formulas and reading from table \ref{tab: cohomology of BO} which classes lift to integral homology, we can compute the third page of the spectral sequence.

\begin{center}
\begin{tikzcd}[row sep={2.5em,between origins}]
\Z & \dots &&&&\\
0               &  0                &   0               & \dots  && \\
0    &0      & \Z_2      & \dots      &  &   \\
0    & \Z_2      & 0      & 0      & \dots    &       \\
\Z      & \Z_2      & 0          & \Z_2      & \Z   & \Z_2 \arrow[llluu,"d_3", bend right, swap]
\end{tikzcd}
\end{center}

In this part of the third page there are still two possible third differentials. 
First of all, the $\Z$ in the upper left corner $(0,4)$ could still be hit by the group at $(3,2)$, which is not displayed here.
However, looking back to the second page, we see that the group at this spot is torsion.
Therefore on the third page it is still torsion and since there is no nonzero map from a torsion group to $\Z$ we see that this third differential vanishes.

However, there is still a possible differential from the $(5,0)$-spot to the $(2,2)$-spot, which could be either an isomorphism or zero.
We can conclude that 
\begin{align*}
    \Omega^H_0 = \Z, \quad \Omega^H_1 = \Z_2, \quad \Omega^H_2 = \Z_2, \quad \Omega^H_3 = \Z_2, \quad \Omega_4^H = \Z^2 \oplus \Z_2 \text{ or } \Z^2
\end{align*}
where the last group depends on this third differential and the relevant extension problem.
This indeterminacy will be relevant when computing $SPT_{U(1)_Q \rtimes \Z_2^C}^4$, but in this paper we will only concern ourselves with dimension $3$ and lower.
One can now compare $SPT^{\Z_2^F}_d$ to $SPT^{U(1)_Q \rtimes \Z_2^C}_d$ for $d\leq 3$ looking at the equations \ref{eq:spin bordism}.
For $d < 3$ they are isomorphic, but for $d = 3$ we deduce $SPT_{U(1)_Q \rtimes \Z_2^C}^3 = \Z^2 \oplus \Z_2$, which is not equal to 
\[
SPT_{\Z_2^F}^3 \cong  \Z.
\]

Since the new $\Z$ and $\Z_2$ are on the bottom row of the spectral sequence, it is not hard to construct invariants realizing them.
For example, comparing with the spectral sequence we see that the three-dimensional $\Z_2$-phase can be realized as the invertible topological field theory with partition function
\begin{equation}
\label{eq:partition function}
Z(M,E) = (-1)^{\int_M w_1(E)^3}, 
\end{equation}
where $E \to M$ is the two-dimensional real vector bundle we get by coupling our bosonic $G_b \cong O(2)$-symmetry to a background gauge field.
Mathematically, it is the vector bundle induced by the $H_d$-structure under the natural map $H_d \to O(2)$.
In the next section we will elaborate on these statements and give an explicit representative $(M,E)$ for which this partition function is nontrivial.

The integer invariant is harder to describe, because integral invariants cannot be represented by classical Atiyah-Segal type field theories.
The physically correct way to describe the partition function is by the geometric secondary invariant associated to the primary invariant one dimension up, see \cite[section 21.2]{differentialcohomology} for a pedagogical introduction.
For us, the primary invariant in dimension four is the first Pontryagin number of the bundle $E \to M$.
The associated secondary invariant is the corresponding classical Chern-Simons theory with gauge group $O(2)$ of level $1$.
\footnote{There are two integral refinements of $p_1 \in H^4(BO(2),\R)$ because the torsion subgroup of $H^4(BO(2),\Z)$ is $\Z_2$ generated by $w_1^4$. The Chern-Simons theories corresponding to $w_1^4 + p_1 \in H^4(BO(2),\Z)$ and $p_1 \in H^4(BO(2),\Z)$ might be different. But after the choice of integral $p_1$ there is a unique differential refinement \cite[Theorem 13.1.1]{differentialcohomology} defining the Chern-Simons action.}

Note additionally that in lower dimensions we have a similar phenomena as was described for the free case at the end of section \ref{subsec:Morita equiv}.
Namely, even though the classifications are equal, the natural comparison maps are not isomorphisms.
For example, for dimension $0+1$ the nontrivial invertible field theory realizing the $\Z_2$ is 
\[
Z(M,E) = (-1)^{\int_M w_1(E)}. 
\]
Restricting this theory to spin manifolds with trivial bundle $E$ corresponds to breaking the particle-hole and $U(1)$-symmetry.
Therefore the comparison map $\Z_2 = SPT_{U(1)_Q \rtimes \Z_2^C}^1 \to SPT_{\Z_2^F}^1 = \Z_2$ is trivial.

We finish this section by remarking that the above computations can also be done without assuming the spin-charge relation.
A reasonable choice of internal symmetry group is then $\Z_2^F \times (U(1)_Q \rtimes \Z_2^C)$ resulting in the spacetime structure group $H_d = \Spin(d) \times (U(1)_Q \rtimes \Z_2^C)$.
The second page of the new spectral sequence is identical and the only thing that changes in the above computations is the expression for $d_2$ which now has $\delta = Sq^2$.
The exact result in this case unfortunately strongly depends on whether third differentials are present, but at least one of the two $\Z_2$'s on place $(2,1)$ can be shown to survive to $E_\infty$.
Therefore also in this case $\Omega^{\Spin \times (U(1)_Q \rtimes \Z_2^C)}_3$ contains a $2$-torsion element which is not present in $\Omega^{\Spin}_3 = 0$ and so $SPT^3_{\Z_2^F \times (U(1)_Q \rtimes \Z_2^C)} \ncong SPT^3_{\Z_2^F}$.

\subsection{Explicit generators}

\label{sec:explicit}

To compute explicit invariants, we first provide a few results that are useful tools to study $H(G)$-manifolds for general symmetry group $G$.

\begin{proposition}
\label{prop: exact}
Let $G,H$ be fermionic groups.
If $(-1)^F \neq 1$ in $H$, then there is an exact sequence of topological groups
\[
1 \to G \overset{i}{\to} G \otimes H \overset{\pi}{\to} H_b \to 1,
\]
where we recall that $H_b := H/\langle (-1)^F \rangle$.
If additionally $G$ has a nontrivial grading homomorphism, then this restricts to an exact sequence
\[
1 \to G_{ev} \overset{i}{\to} (G \otimes H)_{ev} \overset{\pi}{\to} H_b \to 1.
\]
\end{proposition}
\begin{proof}
The maps are defined in the obvious way. 
Since $(g_1 \otimes 1)(g_2 \otimes 1) = (g_1g_2 \otimes 1)$, the map $i$ is an injective homomorphism.
$\pi$ is well-defined, since 
\[
\pi((-1)^F g \otimes (-1)^F h) = [(-1)^F h] = [h] = \pi(g \otimes h)
\]
by design.
It is a homomorphism, since the product
\[
(g_1 \otimes h_1)(g_2 \otimes h_2) = ((-1)^F)^{\theta(h_1)\theta(g_2)} g_1 g_2 \otimes h_1 h_2
\]
is mapped to the class of $h_1 h_2$ in $\frac{H}{\langle (-1)^F \rangle}$.
It is clearly surjective.

To show that $\ker \pi = \operatorname{Im} i$, note first that $\pi i = 1$.
Now suppose $g \otimes h \in \ker \pi$. Then $h = 1$ or $h = (-1)^F$.
If $h = (-1)^F$, then 
\[
i((-1)^F g) = (-1)^F g \otimes 1 = g \otimes (-1)^F = g \otimes h
\]
and otherwise $i(g) = g \otimes h$.
Hence $g \otimes h \in \operatorname{Im}i$.
This finishes the proof that the first sequence is exact.

For the second sequence, note that the restriction of $i$ to the even parts is still well-defined because it intertwines the gradings.
It is also still injective and the composition with $\pi$ is trivial.
Now suppose that $g \otimes h \in \ker \pi$ is even. Then $h \in \{1,(-1)^F\}$ is even and therefore so is $g$.
So we still have $\ker \pi = \operatorname{Im} i$.
To show that $\pi$ is still surjective if $G$ is nontrivially graded, pick $[h] \in H/\langle (-1)^F \rangle$. 
If $h$ is even, then so is $1 \otimes h$ and $\pi(1 \otimes h) = h$.
If $h$ is odd, pick an odd element $g \in G$. Then $g \otimes h$ is even and $\pi(g \otimes h) = h$.
This proves the second sequence.
\end{proof}

\begin{corollary}
\label{cor}
If $G$ is a fermionic group with $(-1)^F \neq 1$, then
\[
1 \to \Spin(d) \to H_d(G) \to G_b \to 1
\]
is exact.
\end{corollary}
\begin{proof}
Apply proposition \ref{prop: exact} to the case where $G = \Pin^+(d)$.
\end{proof}

\begin{proposition}
\label{prop: homotopy pullback}
Suppose $(-1)^F \neq 1 \in G$.
Let $\nu: B G_b \to B^2 \Z_2^F$ be the map corresponding to the extension
\[
1 \to \Z_2^F \to G \to G_b \to 1.
\]
Then the diagram

\[
\begin{tikzcd}
BH(G)  \arrow[r] \arrow[d]
& B G_b \arrow[d, "{(\theta,\nu)}"]
\\
BO \arrow[r,"{(w_1,w_2+w_1^2)}"] 
& B \Z_2^T \times B^2 \Z_2^F
\end{tikzcd}
\]

is a homotopy pullback square.
\end{proposition}
\begin{proof}
We first show the square is commutative. 
The two maps to $B\Z_2^T$ agree because they come from the maps of groups $x \otimes g \mapsto \theta(x), \theta(g)$ which are equal because for $H_d(G)$ we only take the even part of the tensor product.
To show that the maps to $B^2 \Z_2^F$ agree, it is sufficient to prove that the pullback of group extensions

\[
\begin{tikzcd}
1 \arrow[r] & \Z_2^F \arrow[r] \arrow[d,equal]& K_d \arrow[r] \arrow[d] & H_d(G) \arrow[r] \arrow[d]& 1
\\
1 \arrow[r] & \Z_2^F \arrow[r] & \operatorname{Pin}^-(d) \arrow[r] & O_d \arrow[r] & 1
\end{tikzcd}
\]

is isomorphic as a group extension to the pullback of extensions

\[
\begin{tikzcd}
1 \arrow[r] & \Z_2^F \arrow[r] \arrow[d,equal]& K_d' \arrow[r] \arrow[d] & H_d(G) \arrow[r] \arrow[d]& 1
\\
1 \arrow[r] & \Z_2^F \arrow[r] & G \arrow[r] & G_b \arrow[r] & 1
\end{tikzcd}
\]

We will use the model $K_d' = H_d(G) \times_{G_b} G$ and show that it fits in the first pullback square.
There is a homomorphism $K_d' \to \Pin^-(d)$ defined by $((x \otimes g_1) \times_{G_b} g_2) \mapsto x$.
Indeed, we compute the product 
\begin{align*}
    ((x \otimes g_1) \times_{G_b} g_2) \cdot ((x' \otimes g_1') \times_{G_b} g_2') = (-1)^{\theta(g_1) \theta(x')} (x x' \otimes g_1 g_1') \times_{G_b} g_2 g_2'
\end{align*}
and note that since $\theta(g_1) = \theta(x)$ this element is mapped to $(-1)^{\theta(x) \theta(x')} x x'$.
This defines an associative product on $\Pin^+(d)$ through which it becomes isomorphic to $\Pin^-(d)$.
The resulting map to $O_d$ agrees with the composition $G_d \to H_d \to O_d$.
So the square is commutative.
To show it is a homotopy pullback square, we apply corollary \ref{cor}. 
Note that the homotopy fibers of both horizontal maps in the square are $B\Spin$ and the maps are compatible.
\end{proof}

We specialize the discussion to the case at hand and conclude the following.

\begin{corollary}
An $H$-structure on a manifold $M$ consists of an orientation, a two-dimensional real vector bundle $E \to M$ and a $\Pin^+$-structure on $E \oplus TM$.
\end{corollary}
\begin{proof}
Given a manifold with stable tangent bundle $TM: M \to BO$, an $H$-structure consists of a map $M \to BH(G)$ and a homotopy between the composition $M \to BH(G) \to BO$ and the tangent bundle.
The homotopy pullback square of \ref{prop: homotopy pullback} reduces to

\[
\begin{tikzcd}
BH(G)  \arrow[r] \arrow[d]
& B O(2) \arrow[d, "{(0,w_2)}"]
\\
BO \arrow[r,"{(w_1,w_2+w_1^2)}"] 
& B \Z_2^T \times B^2 \Z_2^F
\end{tikzcd}
\]

telling us that this is equivalent to a two-dimensional real vector bundle $E: M \to BO(2)$ together with a homotopy between the the two maps $M \to B \Z_2^T \times B^2 \Z_2^F$.
The two maps to $B \Z_2^T$ are the determinant line bundle $\bigwedge^{top} TM: M \to B \Z_2^T$ and the trivial map, so this gives an orientation on $M$.
The two maps to $B^2 \Z_2^F$ correspond to the cohomology classes $w_1(M)^2 + w_2(M) = w_2(M)$ and $w_2(E)$.
Since $w_2(E \oplus M) = w_2(E) + w_2(M)$, a homotopy between these maps $M \to B^2 \Z_2^F$ is a $\Pin^+$-structure on $E \oplus TM$.
\end{proof}

Recall that the three-dimensional $\Z_2$-phase can be realized as the invertible topological field theory with partition function
\[
Z(M,E) = (-1)^{\int_M w_1(E)^3}.
\]
In physics language, if our theory has both charge and particle-hole symmetry, we have to couple the $U(1)_Q \rtimes \Z_2^C$ to a background gauge field and this purely topological partition function gives some information about which instanton sector spacetime is in.

An example of a three-dimensional $H$-manifold on which the partition function is nontrivial is $\R \mathbb{P}^3$ together with $E = \gamma \oplus \underline{\R}$, where $\gamma \to \R \mathbb{P}^3$ is the canonical line bundle.
Indeed, $\R \mathbb{P}^3$ is orientable, $w_2(\R \mathbb{P}^3) = 0$ and $w_2(\gamma \oplus \underline{\R}) = w_2(\gamma) = 0$, so that $E \oplus T \R \mathbb{P}^3$ admits a spin structure.
Using any of these we have
\[
\int_{\R \mathbb{P}^3} w_1(E)^3 = \int_{\R \mathbb{P}^3} w_1(\gamma)^3 \neq 0 \in \Z_2,
\]
since $w_1(\gamma)^3$ is nontrivial in the top cohomology of $\R \mathbb{P}^3$.
So the theory is nontrivial on the spacetime $(\R \mathbb{P}^3,E)$ with any choice of orientation and $\Pin^+$-structure.
Note that the trivial bundle $E$ over $\R \mathbb{P}^3$ would make the partition function equal to $1$ instead.
Therefore, this theory depends on the choice of instanton sector.
In particular, the theory would not be well-defined if it did not have particle-hole and charge symmetry.

\bibliography{biblio.bib}{}
\bibliographystyle{plain}

\subsection{Conflict of Interest}

The author states that there is no conflict of interest.

\subsection{Data Availability}

Data sharing not applicable to this article as no datasets were generated or analysed
during the current study

\end{document}